\documentclass[conference]{IEEEtran}
\IEEEoverridecommandlockouts

\usepackage{cite}
\usepackage{amsmath,amssymb,amsfonts}
\usepackage{amsthm}
\usepackage{graphicx}
\usepackage[normalem]{ulem}
\usepackage{tikz}
\usepackage{circuitikz}
\usepackage{pgfplots}
\usepackage{booktabs}
\usepackage{enumitem}
\usepackage{hyperref}
\pgfplotsset{compat=1.18}
\usepackage{textcomp}
\usepackage{xcolor}
\usepackage{float, comment}

\def\BibTeX{{\rm B\kern-.05em{\sc i\kern-.025em b\kern-.08em
T\kern-.1667em\lower.7ex\hbox{E}\kern-.125emX}}}

\theoremstyle{definition}
\newtheorem{definition}{Definition}
\newtheorem{example}{Example}
\newtheorem{remark}{Remark}

\theoremstyle{plain}
\newtheorem{theorem}{Theorem}
\newtheorem{lemma}{Lemma}
\newtheorem{proposition}{Proposition}
\newtheorem{corollary}{Corollary}

\title{Distortion-Rate Outage Probability for Non-ergodic Classes of Sources}
\author{
\IEEEauthorblockN{
Dongmin Park, 
 Ibrahim Abou-Faycal\\ 
 \textit{Department of Electrical and Computer Engineering,}\\ \textit{American University of Beirut}, Beirut, Lebanon \\
  Emails: djp00@aub.edu.lb, Ibrahim.Abou-Faycal@aub.edu.lb  \and Razane Tajeddine \\ \textit{Department of Mathematics and Systems Analysis,} \\
  \textit{Aalto University}, Finland \\
  Email: rstajeddine@gmail.com }\thanks{This work was done while R.T. was affiliated with the department of Electrical and Computer Engineering at the American University of Beirut.}} 

\begin{document}

\maketitle


\begin{abstract}
    The concept of outage has been commonly used in the context of transmission rates over non-ergodic channels. For those channels, one channel state or realization may be encountered for the whole transmission duration of a codeword. In the dual setting of source coding and whenever source statistics are not ergodic, we define in this work the notion of outage; we discuss related properties, consider the special case of compact source classes and provide examples.
\end{abstract}

\section{Introduction}

Shannon's channel coding theorem establishes the channel capacity as the supremum of the rates at which information could be transmitted reliably over a channel~\cite{cover_elements_2005}. Whenever the transmission rate exceeds the capacity, reliable communication is impossible. In wireless communication settings, fading models such as Rayleigh fading~\cite{923716} are often adopted for channels. When the fading is slow, the channel state remains fixed over the entire codeword, making the channel non-ergodic. In this case, one cannot simply average over the channel states and the classical capacity loses its operational meaning. For such channels, the notion of \emph{outage capacity} was introduced (see, \emph{e.g.},~\cite{Tse_Viswanath_2005}), in which a prescribed probability of transmission failure is tolerated. With the overall communication chain in mind, an expected distortion may not be the appropriate metric to characterize reliability when using such non-ergodic channels. To that effect, the notion of distortion outage for MIMO block fading channel was introduced in \cite{5513482}.

When it comes to source coding, for a source with a known distribution $p_\theta$, Shannon's rate-distortion function $R_\theta(D)$~\cite{https://doi.org/10.1002/j.1538-7305.1948.tb01338.x} gives the infimum of the  compression rates for which the expected distortion between the source and its reconstruction does not exceed $D$. In this work, we consider a class of sources $\Lambda$, where each $\theta \in \Lambda$ is associated with its own probability law $p_\theta(\cdot)$. The encoder and decoder have no knowledge of $\theta$, \emph{i.e.}, they are \emph{blind}. Such coding schemes were studied in~\cite{SAKRISON1969165, ZivJ.1972Cosw}.

In this work, we study the source coding analogue of outage capacity: as with fading channels, we distinguish two regimes according to how $\theta$ evolves over time. In the \emph{fast} regime, $\theta$ is redrawn on a short time scale according to an ergodic process over $\Lambda$, so the resulting symbol sequence is itself ergodic. One can then average over $\Lambda$, describe the class by a single law and then use the rate-distortion theory tools directly. In the \emph{slow} regime, a single $\theta\in\Lambda$ is drawn once and held fixed for the entire transmission of the codeword. The symbols are then generated Independent and Identically Distributed (IID) according to $p_\theta$. The resulting process is a mixture of ergodic components and is non-ergodic in general. So, again, averaging over $\Lambda$ loses its operational meaning. Moreover, the realized $\theta$ may require a rate that exceeds the available rate budget to meet $D$. We refer to this event as an \emph{outage}. Building on Sakrison's rate-distortion function for a class of sources~\cite{SAKRISON1969165}, we formally define the \emph{Distortion-Rate Outage Probability} (DROP) for this setting.

\begin{figure}[!t]
\centering
\resizebox{\linewidth}{!}{%
\begin{circuitikz}
\tikzstyle{every node}=[font=\Large]

\draw (-10.75,5.5) to[short] (-8.25,5.5);
\node [font=\Huge] at (-10.5,5.0) {Source};
\node [font=\Huge] at (-9.5,6) {$X^n$};

\draw (-8.25,6.75) rectangle
    node {\huge Encoder \(f_n\)}
    (-4,4.25);

\draw (-4,5.5) to[short] (1.75,5.5);
\node [font=\Huge] at (-1.25,6) {index};
\node [font=\huge] at (-1.15,4.9) {$\in\!\{1,2,\ldots,2^{nR}\}$};

\draw (6,5.5) -- (8.75,5.5);
\node [font=\Huge] at (7.5,6.1) {$\hat{X}^n$};

\node [font=\Huge] at (8.5,4.85) {Output};

\draw (1.75,6.75) rectangle
    node {\huge Decoder \(g_n\)}
    (6,4.25);

\end{circuitikz}%
}
\vspace{-0.6cm}
\caption{Source coding.
\label{fig:source_coding}}
\end{figure}
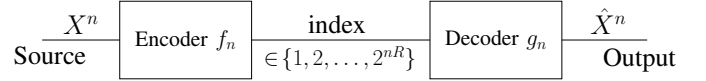

\section{Preliminaries}\label{sec:preliminaries}

Consider the classical source coding problem, shown in Figure~\ref{fig:source_coding}. A source \(\theta\) generates a sequence of \(n\) IID symbols \(X^n = (X_1,\cdots,X_n) \in \mathcal X^n\), where \(\mathcal X\) is a separable metric space. This sequence is encoded by \(f_n:\mathcal X^n\to\{1,\ldots,M_n\}\) and decoded by \(g_n:\{1,\ldots,M_n\}\to \hat{\mathcal X}^n\), yielding the reconstruction \(\hat X^n := g_n( f_n( X^n ))\). 
Let \(d:\mathcal X\times\hat{\mathcal X}\to[0,\infty)\) be a per-symbol distortion measure, where \(d(x_i,\hat x_i)\) quantifies the distortion between a source symbol \(x_i\) and its reconstruction \(\hat x_i\). The distortion between a source sequence \(x^n\) and its reconstruction \(\hat x^n\) is defined as the per-symbol average, \(d(x^n,\hat x^n):=\frac1n\sum_{i=1}^nd(x_i,\hat x_i)\). The definition of an encoder--decoder pair as a code, together with its rate and achievability are formalized in Section~\ref{sec:outage}.

 The same schematic as in Figure~\ref{fig:source_coding} applies whether there is a single source, or the code is designed for a class \(\Lambda\) of sources, in which case \((f_n,g_n)\) is expected to work well across the entire class in the sense of not exceeding a prescribed distortion level.

For a source with known probability law \(p_\theta\), Shannon's information rate--distortion function \cite{cover_elements_2005} is defined as
\begin{equation*}
    R_\theta(D) := \underset{p_{\hat X|X}\,:\,\mathbb{E}_\theta[d(X,\hat X)]\le D}{\inf} I_\theta(X;\hat X),
\end{equation*}
where the minimization is over all conditional distributions \(p_{\hat X|X}\) such that the expectation, taken with respect to the joint law \(p_\theta(x)\,p_{\hat X|X}(\hat x|x)\), satisfies the fidelity constraint \(D\). The Shannon coding theorem establishes the operational meaning of \(R_\theta(D)\) as the infimum of the rates at which a source with law \(p_\theta\) can be compressed such that the average per-symbol distortion between the source and its reconstruction does not exceed \(D\).

We exclusively consider in this work a class \(\Lambda\) of sources such that \(\mathbb E_\theta[d(X,x_0)]<\infty\) for all \(\theta\in\Lambda\), for some \(x_0\in\mathcal X'\). Sakrison \cite{SAKRISON1969165} extended this notion to such a class \(\Lambda\) of sources by providing a definition for the information rate--distortion function of the class, \(R_\Lambda^I(D)\), via a minimax formulation:
\begin{equation}\label{eq:Sakrison_RLambda}
    R_\Lambda^I(D) := \inf_{p_{\hat X|X}\in Q_\Lambda(D)} \sup_{\theta\in\Lambda} I_\theta(X;\hat X),
\end{equation}
where
\begin{equation*}
    Q_\Lambda(D) := \left\{ p_{\hat X|X} : \mathbb{E}_\theta[d(X,\hat X)] \le D \quad \forall\,\theta\in\Lambda \right\}.
\end{equation*}

That is, the optimization is over the choice of a single conditional distribution \(p_{\hat X|X}\), satisfying the distortion constraint for \emph{every} source in the class simultaneously, which minimizes the worst-case mutual information over \(\Lambda\). 

Following Sakrison \cite{SAKRISON1969165}, a class \(\Lambda\) is said to be \emph{compact} if, given any \(\varepsilon_1>0\), there exists a totally bounded set \(\mathcal X'\subseteq\mathcal X\) (totally bounded in the sense that \(\mathcal X'\) can be covered by finitely many neighborhoods of any given radius) and a function 
\begin{equation*}
    x_p: \mathcal{X} \longrightarrow \mathcal{X}'
\end{equation*}
such that
\begin{equation}\label{eq:compactness}
    \mathbb E_\theta[d(X,x_p(X))] \le \varepsilon_1 \qquad \forall\,\theta\in\Lambda.
\end{equation}

Intuitively, compactness requires that \emph{all} sources in \(\Lambda\) have most of their probability mass concentrated on a \emph{common} totally bounded set. A special case of compactness is when there is a finite number of sources in the class.

Additionally, Sakrison proves a source coding theorem \cite{SAKRISON1969165} in two parts:
\begin{itemize}
\item The \emph{negative} (converse) statement holds for \emph{any} class \(\Lambda\), compact or not: if an encoder--decoder pair achieves distortion at most \(D\) for every \(\theta\in\Lambda\), the rate of the code must be at least \(R_\Lambda^I(D)\). Therefore, \(R_\Lambda^I(D)\) is always a necessary rate for uniform-distortion coding across the class. 
\item The \emph{positive} (achievability) statement requires the class to be \emph{compact}, and states given any positive \(D\) and \(\delta\), there exists a code of sufficiently large length achieving a distortion at most \(D+\delta\) for every \(\theta\in\Lambda\), with rate no greater than \(\sup_{\theta\in\Lambda} R_\theta(D) + \delta\).
\end{itemize}

Consequently and in general,
\begin{equation}\label{eq:general_ineq}
    R_\Lambda^I(D) \ge \sup_{\theta\in\Lambda} R_\theta(D),
\end{equation}
which follows directly from the definition of \(R_\Lambda^I(D)\): restricting a single conditional distribution to satisfy the constraint for every \(\theta\in\Lambda\) can only make the worst-case rate at least as large as the rate for the individual worst-case source. Combining the negative statement (giving \(R_\Lambda^I(D)\) as a necessary rate) with the positive statement for a compact class (proving that a rate arbitrarily close to  \raisebox{-0.01in}{\(\sup_{\theta\in\Lambda} R_\theta(D)\)} is achievable), together with \eqref{eq:general_ineq}, forces equality:
\begin{equation}\label{eq:Sakrison_compact_sup}
    R_\Lambda^I(D) = \sup_{\theta\in\Lambda} R_\theta(D) \qquad \text{for a compact class } \Lambda.
\end{equation}

As noted by Sakrison, this equality is of considerable practical importance, since \(R_\Lambda^I(D)\) as defined in \eqref{eq:Sakrison_RLambda} is generally difficult to evaluate directly, whereas \(\sup_{\theta\in\Lambda} R_\theta(D)\) can often be computed from the individual rate--distortion functions of the sources in the class. For a non-compact class, only the necessary-rate direction is guaranteed; \(R_\Lambda^I(D)\) need not be achievable (a notion we formalize in Section~\ref{sec:outage}), and the true minimum required rate may exceed it.

Similarly to the case of a single source, \(R_\Lambda^I(D)\) is nonincreasing and convex in \(D\) \cite{SAKRISON1969165}, one can also define the information distortion-rate function for any class of sources, denoted \(D_\Lambda^I(R)\), as the ``inverse'' of \(R_\Lambda^I(D)\):
\begin{align*}
    D_\Lambda^I(R) &:= \inf\{D : R_\Lambda^I(D)\le R\}\\
    R_\Lambda^I(D) &:= \inf\{R : D_\Lambda^I(R)\le D\}.
\end{align*}

Additionally, if \(R_\Lambda^I(0) = R_0\) is finite, then \(D_\Lambda^I(R) = 0\) for all \(R > R_0\), and vice versa. We use this same notion for ``inverses" for the operational objects that we introduce in the following section.

\section{Distortion-Rate Outage Probability}\label{sec:outage}


In this section, we first formalize the operational objects introduced informally in Section~\ref{sec:preliminaries} and restate Sakrison's source coding theorem using our notation. We emphasize that we are primarily interested in a \emph{blind} scenario where neither the encoders nor the decoders have any information about the chosen source \(\theta\). 

In what follows, we consider a class \(\Lambda\) of sources such that \(\mathbb E_\theta[d(X,x_0)]<\infty\) for all \(\theta\in\Lambda\) and for some \(x_0\in\mathcal X'\). The sources $\theta \in \Lambda$ are chosen with a probability law generically denoted $\Pr(.)$.

\begin{definition}[Code]\label{def:code}
A  blocklength-$n$ \emph{code} \( C_n = (f_n,g_n) \) is a pair consisting of an encoder \(f_n:\mathcal X^n\to\{1,\ldots,M_n\}\) and a decoder \(g_n:\{1,\ldots,M_n\}\to\hat{\mathcal X}^n\). The \emph{rate} of \(C_n\) is
\begin{equation}
    R(C_n) := \frac1n \log M_n
\end{equation}
and, for a source \(\theta\in\Lambda\), the \emph{distortion} of \(C_n\) is
\begin{equation}
    d_\theta(C_n) := \mathbb E_\theta\big[d\big(X^n,\,g_n(f_n(X^n))\big)\big].
\end{equation}
\end{definition}

\begin{definition}[Achievability]\label{def:achievable}
A pair \((R,D)\) is \emph{achievable} for the class \(\Lambda\) if there exists a sequence of code \(\{C_n\}_{n}\) with
\begin{equation*}
    \limsup_{n\to\infty} R(C_n) \le R \hspace{10pt} \& \hspace{10pt}
    \limsup_{n\to\infty} d_\theta(C_n) \le D \quad \forall\,\theta\in\Lambda.
\end{equation*}
The \emph{achievable region} \(\mathcal R_\Lambda\) is the closure of the set of all achievable pairs.
\end{definition}

\begin{definition}[Rate--distortion function]\label{def:R_op}
For the class \(\Lambda\),
\begin{align*}
    & R_\Lambda(D) := \inf\{R : (R,D)\in\mathcal R_\Lambda\} \\
    \& \quad & D_\Lambda(R) := \inf\{D : (R,D)\in\mathcal R_\Lambda\}.
\end{align*}

\end{definition}

By construction, \raisebox{-0.01in}{\(R_\Lambda(D)\)} is well-defined for \(\Lambda\) and carries its intended operational meaning directly: it is the infimum of the rates at which a single code, designed without knowledge of \(\theta\), compresses any source in \(\Lambda\) with asymptotic distortion at most \(D\) for every \(\theta\in\Lambda\).

\begin{theorem}[Sakrison's source coding theorem \cite{SAKRISON1969165}]\label{thm:sakrison}
For the class \(\Lambda\),
\begin{equation}\label{eq:op_vs_I}
    R_\Lambda(D) \ge R_\Lambda^I(D)\ge \sup_{\theta\in\Lambda}R_{\theta}(D).
\end{equation}
If \(\Lambda\) is compact, then
\begin{equation}\label{eq:op_vs_I_compact}
    R_\Lambda(D) = R_\Lambda^I(D) = \sup_{\theta\in\Lambda}R_\theta(D).
\end{equation}
\end{theorem}

The remainder of this section is stated in terms of the operational quantity \(R_\Lambda(D)\),
which always carries the intended operational guarantee; \eqref{eq:op_vs_I_compact} is invoked wherever a compact subclass allows these to be evaluated via \(\sup_{\theta} R_\theta(D)\). 

Consider a \emph{slow} class of sources: for the whole duration of the transmitted codeword, the source produces characters according to a distribution \(\theta\in\Lambda\) chosen according to the probability law on \(\Lambda\). Suppose the distortion level should be below a threshold \(D\) for every source (with probability 1). Then the minimum required rate is \(R_\Lambda(D)\). However, similarly to the notion of outage for channel coding, if we allow the distortion level to exceed this threshold with some probability, we may require a lower rate \(R\). A \emph{distortion outage} occurs when the distortion level exceeds the threshold \(D\). Due to the dual nature between rate and distortion, this problem may be viewed from another angle: one may instead want to limit the rate by a threshold \(R\), resulting in a distortion level \(D_\Lambda(R)\). However, if we allow some probability for the rate to exceed \(R\), we may achieve lower distortion. We call such events \emph{rate outage}.

Let \(\Sigma\subseteq\Lambda\) denote a subclass of \(\Lambda\). For a rate limit \( R\) and distortion level \(D\), consider all subclasses such that
\begin{equation*}
    \{\Sigma\subseteq\Lambda : R_\Sigma(D)\le R\}.
\end{equation*}
We wish to find a subclass maximizing the probability
\begin{equation*}
    \sup_{\Sigma\subseteq\Lambda\,:\,R_\Sigma(D)\le R} \Pr(\Sigma).
\end{equation*}
Notice that the optimization constraint is over the rate--distortion functions of the subclasses \(\Sigma\) rather than of individual sources \(\theta\). We can now formally define the rate outage probability and distortion outage probability. Throughout, the optimization is understood to range over \(\Pr\)-measurable subclasses \(\Sigma\subseteq\Lambda\).

\begin{definition}
The \emph{rate outage probability} for a rate limit \(R\) at a given distortion level \(D\) is defined as
\begin{align}
    P_{\rm out}^{R}(R,D) &= 1-\sup_{\Sigma\subseteq\Lambda\,:\,R_\Sigma(D)\le R}\Pr(\Sigma)\label{eq:opt_rate_outage}
\end{align}
\end{definition}

\begin{definition}
The \emph{distortion outage probability} for a threshold distortion level \(D\) at a given rate \(R\) as
\begin{align}
    P_{\rm out}^D(D,R) &= 1-\sup_{\Sigma\subseteq\Lambda\,:\,D_\Sigma(R)\le D} \Pr(\Sigma)\label{eq:opt_distortion_outage}
\end{align}
\end{definition}

Note that these two definitions are equivalent, in part due to the dual nature of rate-distortion and distortion-rate functions, as shown in proposition \ref{prop:countable_atoms}. 


Define the following sets of subclasses:
\begin{align}
    \mathcal{A}_D:=\{\Sigma\subseteq\Lambda : D_\Sigma(\bar{R}) \le \bar{D}\},\\
    \mathcal{A}_R:=\{\Sigma\subseteq\Lambda : R_\Sigma(\bar{D}) \le \bar{R}\}.
\end{align}
If \(\Sigma\in\mathcal{A}_R\), then \(R_\Sigma(\bar{D}) \le \bar{R} \implies D_\Sigma(\bar{R})=\inf\{D:R_\Sigma(D) \le \bar{R}\}\le\bar{D} \implies \Sigma\in\mathcal A_D\), so \(\mathcal{A}_R\subseteq\mathcal{A}_D\), and by the same reasoning, \(\mathcal A_D\subseteq\mathcal A_R\). Therefore, \(\mathcal{A}_D=\mathcal{A}_R\), and the optimization problems in equations~\eqref{eq:opt_rate_outage},\eqref{eq:opt_distortion_outage} are taken over the same set of subclasses, and therefore the probabilities in~\eqref{eq:opt_rate_outage} and \eqref{eq:opt_distortion_outage} are equal. Hereafter, we drop the superscript and simply denote either quantity as
\begin{equation}\label{eq:pout_rate_distortion}
    P_{\rm out}(R,D):=P_{\rm out}^R(\bar{R},\bar{D})=P_{\rm out}^D(\bar{D},\bar{R}),
\end{equation}
where  \(R=\bar R\) and \(D=\bar D\). We call \(P_{\rm out}(\cdot,\cdot)\) the \emph{distortion-rate outage probability} of the class \(\Lambda\) of sources.

\begin{proposition}\label{prop:countable_atoms}
For each fixed $D$, there are at most countably many values $\bar R \ge 0$ such that
\[
  \Pr_{\theta\in\Lambda}\left(R_\theta(D) = \bar R\right) > 0.
\]
\end{proposition}
\begin{proof}
The sets  $A_{\bar R} := \{\theta : R_\theta(D) = \bar R\}$ are pairwise disjoint for distinct $\bar R$. Since $\Pr$ is a probability measure, for each $n \in \mathbb{N}$, fewer than $n$ of them can satisfy $\Pr(A_{\bar R}) > 1/n$, as otherwise their union would have probability exceeding 1. Therefore, \[\{\bar R : \Pr(A_{\bar R}) > 0\} = \bigcup_{n} \{\bar R : \Pr(A_{\bar R}) > 1/n\}\] is a countable union of finite sets, and hence at most countably many $\bar R$ carry positive mass.
\end{proof}

\begin{remark}\label{rem:boundary_convention}
By Proposition~\ref{prop:countable_atoms}, for any fixed \(D\), the boundary event \(\{\theta:R_\theta(D)=\bar R\}\) has a \(\Pr\)-probability zero for all but at most countably many values of \(\bar R\). Throughout this paper, all inequalities involving a rate budget \(\bar R\) (e.g., \(\Sigma^{\rm opt}(\bar R,D)=\{\theta:R_\theta(D)\le\bar R\}\), 
are stated with non-strict inequality. Whenever \(\bar R\) is chosen outside this countable exceptional set, this choice is immaterial: \(\Pr(\{\theta:R_\theta(D)\le\bar R\})=\Pr(\{\theta:R_\theta(D)<\bar R\})\), and \(P_{\rm out}(\bar R,D)\) is unambiguous regardless of convention. We do not further resolve the exceptional case, in which the precise treatment of the boundary reflects a definitional choice rather than an operationally distinguishable guarantee, consistent with Theorem~\ref{thm:sakrison}, which is itself only stated with an arbitrarily small margin \(\delta>0\).
\end{remark}



\section{Computing the Outage Probability for Compact Classes}\label{sec:finite}

In general, evaluating \(P_{\rm out}(R,D)\) requires optimizing over all subclasses \(\Sigma\subseteq\Lambda\), and \(R_\Sigma(D)\) is itself generally intractable to compute for an arbitrary \(\Sigma\) unless \(\Sigma\) is compact (Theorem~\ref{thm:sakrison}). We now show that when \(\Lambda\) itself is compact, this optimization complexity reduces considerably, and \(P_{\rm out}(R,D)\) can be computed directly from the rate--distortion functions of the individual sources \(\theta\in\Lambda\).

\begin{lemma}\label{lem:subclass_compact}
If  \raisebox{-0.01in}{\(\Lambda\)} is compact, then every subclass  \raisebox{-0.01in}{\(\Sigma\subseteq\Lambda\)} is compact.
\end{lemma}
\begin{proof}

Since $\Lambda$ is compact, for any $\epsilon_1>0$, there exists a totally bounded set $\mathcal X' = \mathcal{X}'(\epsilon_1)$ and a function \(x_p(\cdot)\) with range \(\mathcal X'\) such that \[\mathbb E_\theta[d(X,x_p(X))]\le\epsilon_1\] for all \(\theta\in\Lambda\).

Since \(\Sigma\subseteq\Lambda\), the same \(\mathcal X'\) and \(x_p\) satisfy \[\mathbb E_\theta[d(X,x_p(X))]\le\epsilon_1\] for all \(\theta\in\Sigma\). As $\epsilon_1$ is arbitrary, \(\Sigma\) is compact.
\end{proof}

Consequently, by Theorem~\ref{thm:sakrison}, \(R_\Sigma(D)=\sup_{\theta\in\Sigma}R_\theta(D)\) for \emph{every} subclass \(\Sigma\subseteq\Lambda\) and not only for \(\Lambda\) itself. This allows us to reduce the subclass-level optimization defining \(P_{\rm out}(R,D)\) to a pointwise condition on individual sources.

\begin{theorem}\label{thm:compact_reduction}
Let \(\Lambda\) be compact. Then the subclass maximizing \(\Pr(\Sigma)\) subject to \(R_\Sigma(D)\le R\) is
\begin{equation}\label{eq:sigma_opt}
    \Sigma^{\rm opt} := \{\theta\in\Lambda : R_\theta(D)\le R\},
\end{equation}
and the outage probability is
\begin{align}
    P_{\rm out}(R,D) &= 1-\Pr(\Sigma^{\rm opt}) \notag \\
    &= 1-\int_{\{\theta\in\Lambda\,:\,R_\theta(D)\le R\}} d\Pr(\theta),\label{eq:pout_compact}
\end{align}
provided \(\Sigma^{\rm opt}\) is \(\Pr\)-measurable.
\end{theorem}
\begin{proof}
By Lemma~\ref{lem:subclass_compact}, \(R_\Sigma(D)=\sup_{\theta\in\Sigma}R_\theta(D)\) for every \(\Sigma\subseteq\Lambda\), so the constraint \(R_\Sigma(D)\le R\) holds if and only if \(R_\theta(D)\le R\) for every \(\theta\in\Sigma\), i.e., if and only if \(\Sigma\subseteq\Sigma^{\rm opt}\). Since \(\Pr\) is monotone under inclusion, \(\Pr(\Sigma)\le\Pr(\Sigma^{\rm opt})\) for every feasible \(\Sigma\), with equality at \(\Sigma=\Sigma^{\rm opt}\). Hence \(\Sigma^{\rm opt}\) is the (essentially unique) maximizer, and \eqref{eq:pout_compact} follows from the definition of \(P_{\rm out}\) in Section~\ref{sec:outage}.
\end{proof}

\begin{corollary}[Finite classes]\label{cor:finite}
If \(\Lambda\) is finite (and hence compact), \eqref{eq:pout_compact} reduces to
\begin{equation}\label{eq:pout_finite}
    P_{\rm out}( R,D)
    =
    1-\sum_{\theta\in\Lambda\,:\,R_\theta(D)\le R}\Pr(\theta).
\end{equation}
\end{corollary}

Corollary~\ref{cor:finite} provides a direct computational recipe: given the individual rate--distortion functions \(R_\theta(D)\) and the prior $\Pr(\theta)$ for each $\theta\in\Lambda$, the outage probability at any $(R,D)$ is obtained by summing \(\Pr(\theta)\) over the sources whose individual rate--distortion function does not exceed \(R\) at distortion \(D\); No explicit search over subclasses, and no evaluation of \(R_\Sigma(D)\) for any \(\Sigma\), is required. More generally, Theorem~\ref{thm:compact_reduction} shows this same reduction holds for any compact class, finite or not such as a bounded-variance Gaussian class, with the sum replaced by an integral against the prior \(\Pr\) over \(\Sigma^{\rm opt}\).

\begin{remark}[Side Information]
Suppose we have an arbitrary \emph{compact} class, \(\Lambda\), and both the encoder and decoders are \emph{informed} of the source \(\theta\) generating the symbols, \emph{i.e.}, with perfect knowledge of the source statistics. Then, the distortion-rate outage probability is \(1-\sum_{\theta\in\Lambda\,:\,R_\theta(D)\le R}\Pr(\theta)\).
We could also consider the case where only the encoder is informed of the source \(\theta\), and \(\Lambda\) is finite (and therefore compact). Here, we may use a finite number of prefix bits to signal \(\theta\). This essentially informs the decoder of the source \(\theta\). Even in this case, the outage probability remains the same: the prefix cost is a vanishing fraction of the rate as the block length grows. Therefore, this neither relaxes nor tightens which sources satisfy \(R_\theta(D)\le R\). Therefore, it does not the distortion-rate outage probability asymptotically.

Although no improvements are feasible in terms of DROP, side information may still be useful in the following cases:
\begin{enumerate}[label=(\roman*)]
    \item Finite blocklength \(n\) performance
    \item Non-compact classes when \(R_\Lambda(D)\neq\sup_{\theta\in\Lambda} R_\theta(D)\)
\end{enumerate}
\end{remark}

We illustrate the distortion-rate outage probability \(P_{\rm out}(R,D)\) with two examples: a finite class of Bernoulli sources under Hamming distortion, for which Corollary~\ref{cor:finite} applies directly, and a continuum class of Gaussian sources under mean-squared error, which is known to be compact.

\begin{example}[A finite class of Bernoulli sources]

Let \(\Lambda=\{p_1,\ldots,p_k\}\subset(0,1)\), where each \(\theta_i=p_i\) is a Bernoulli\((p_i)\) source, and let \(d(\cdot,\cdot)\) be the Hamming distortion measure. The rate--distortion function of a Bernoulli\((p)\) source under Hamming distortion \cite{cover_elements_2005} is
\begin{equation*}
    R_\theta(D)= 
    \begin{cases} 
        H_b(p) - H_b(D), & 0\le D \le \min \{p,1-p\},\\ 
        0, & D > \min \{p, 1-p\},
    \end{cases}
\end{equation*}
where \(H_b(\cdot)\) is the binary entropy function. Since \(\Lambda\) is finite, by Corollary~\ref{cor:finite}, for a rate limit \(R\) and distortion level \(D\),
\begin{equation*}
P_{\rm out}(R,D) = 1-\sum_{p_i\in\Sigma^{\rm opt}}\Pr(p_i),
\end{equation*}
where
{\small\begin{equation*}
    \Sigma^{\rm opt}=
    \begin{cases}
        \{p_i\in\Lambda : H_b(p_i)\le R+H_b(D)\},&\text{when } 0<D<0.5,\\
        \Lambda, & \text{otherwise}.
    \end{cases}
\end{equation*}}

As a concrete instance, let \(\Lambda=\{0.1,0.3,0.5\}\) with prior \(\Pr(0.1)=0.5,\ \Pr(0.3)=0.3,\ \Pr(0.5)=0.2\), and fix \(D=0.05\). Here, \(H(D)\approx0.286\) bits. Table~\ref{tab:bernoulli} lists \(R_{p_i}(D)\).

\begin{table}[H]
\centering
\caption{Bernoulli source classes at distortion level \(D=0.05\).}
\label{tab:bernoulli}
\small
\begin{tabular}{c c c c}
\toprule
\(p_i\) & \(\Pr(\theta=p_i)\) & \(H(p_i)\) & \(R_{p_i}(D)\) \\
& & (bits) & (bits/sample) \\
\midrule
0.1 & 0.5 & 0.469 & 0.183 \\
0.3 & 0.3 & 0.881 & 0.595 \\
0.5 & 0.2 & 1.000 & 0.714 \\
\bottomrule
\end{tabular}
\end{table}

For \(R<0.183\), no source satisfies \(R_{p_i}(D)\le R\), so \(\Sigma^{\rm opt}=\emptyset\) and \(P_{\rm out}(R,0.05)=1\). For \(0.183\le R<0.595\) (e.g., \(R=0.2\)), \(\Sigma^{\rm opt}=\{0.1\}\) and \(P_{\rm out}(0.2,0.05)=1-0.5=0.5\). For \(0.595\le R<0.714\) (e.g., \(R=0.65\)), \(\Sigma^{\rm opt}=\{0.1,0.3\}\) and  \raisebox{-0.01in}{\(P_{\rm out}(0.65,0.05)=1-0.8=0.2\)}. For \(R\ge0.714\),  \raisebox{-0.02in}{\(\Sigma^{\rm opt}=\Lambda\)} and \(P_{\rm out}=0\).

In Figure~\ref{fig:pout_bernoulli}, we show \(P_{\rm out}(R,D)\) for the tighter distortion level, \(D=0.01\), where \(H(0.01)\approx0.081\) bits. We notice that since \(R_{p_i}(D)\) is decreasing in \(D\), every breakpoint shifts to a higher rate, and the entire curve is shifted rightward relative to the higher distortion level, \(D=0.05\).

\begin{figure}[!ht]
\centering
\begin{tikzpicture}
\begin{axis}[
    width=\linewidth, height=6cm,
    xlabel={$R$ (bits)}, ylabel={$P_{\rm out}(R,D)$},
    xmin=0, xmax=1.0, ymin=0, ymax=1.1,
    grid=both, grid style={gray!20},
    scaled x ticks=false,
    x tick label style={/pgf/number format/fixed, /pgf/number format/precision=3}, ytick={0,0.2,0.5,1.0},
    legend pos=north east, legend cell align=left,
]
\addplot+[const plot, mark=none, thick] coordinates {
(0,1) (0.183,0.5) (0.595,0.2) (0.714,0)
};
\addlegendentry{$D=0.05$}
\addplot+[const plot, mark=none, thick, dashed] coordinates {
(0,1) (0.388,0.5) (0.801,0.2) (0.919,0)
};
\addlegendentry{$D=0.01$}
\end{axis}
\end{tikzpicture}
\caption{Outage probability \(P_{\rm out}(R,D)\) for the finite Bernoulli class, as a function of rate budget \(R\), for two distortion levels \(D\in\{0.01,0.05\}\). The looser target \(D=0.05\) shifts every drop to a lower rate, since \(R_{p_i}(D)\) decreases as \(D\) increases.}
\label{fig:pout_bernoulli}
\end{figure}
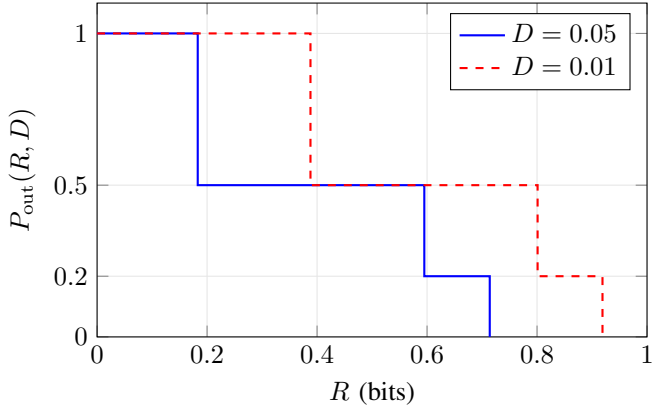

\end{example}

\begin{example}[A continuous class of Gaussian sources]

Let \(\Lambda = \{\theta=(\mu,\sigma^2) : |\mu|\le M,\ 0<\sigma^2\le\sigma_0^2\}\) and let \(d(x,\hat x)=(x-\hat x)^2\). This class is compact according to~\cite[Section 6]{SAKRISON1969165}. By  Theorem~\ref{thm:compact_reduction}, \(R_\Lambda(D)=\sup_{\theta\in\Lambda}R_\theta(D)\). The same holds for every subclass \(\Sigma\subseteq\Lambda\). With MSE rate--distortion being invariant to the mean, \[R_\theta(D)=\tfrac12\log_2(\sigma^2/D),\] for \(D\le\sigma^2\), depending on \(\theta=(\mu,\sigma^2)\) only through \(\sigma^2\). Consequently,
\begin{equation*}
    \Sigma^{\rm opt} = \{(\mu,\sigma^2)\in\Lambda : \sigma^2 \le D\cdot 2^{2R}\},
\end{equation*}
and is an interval constraint on the variance alone.

Suppose \(\mu\) and \(\sigma^2\) are independent under \(\Pr\), with a uniform \( \sigma^2 \sim \mathcal{U}(0,\sigma_0^2]\) (and \(\mu\) arbitrarily distributed on \([-M,M]\)). Then
\begin{equation*}
    \Pr(\Sigma^{\rm opt}) = \Pr\big(\sigma^2\le D\cdot2^{2R}\big) = \frac{\min \{ D \cdot 2^{2R}, \,\sigma_0^2\}}{\sigma_0^2}.
\end{equation*}
By Theorem~\ref{thm:compact_reduction},
\begin{equation}\label{eq:gaussian_pout}
    P_{\rm out}(R,D) = \max \left\{ 1-\frac{D\cdot2^{2R}}{\sigma_0^2}, 0 \right\}.
\end{equation}

Equation~\eqref{eq:gaussian_pout} shows that \(P_{\rm out}\) is decreasing in \(R\) and \(D\), reaching zero once \(R\ge\tfrac12\log_2(\sigma_0^2/D)=\sup_{\theta\in\Lambda}R_\theta(D)\). This is consistent with \eqref{eq:general_ineq}: no outage is possible once the rate budget covers the worst-case source in the class.

\begin{figure}[!ht]
\centering
\begin{tikzpicture}
\begin{axis}[
    width=\linewidth, height=6cm,
    xlabel={\(R\) (bits)}, ylabel={\(P_{\rm out}(R,D)\)},
    xmin=0, xmax=2, ymin=0, ymax=1,
    legend pos=north east, legend cell align=left,
    grid=both, grid style={gray!20},
]
\addplot+[mark=none, thick] table[x=R,y=D01] {
R D01
0.00 0.9000
0.10 0.8851
0.20 0.8680
0.30 0.8484
0.40 0.8259
0.50 0.8000
0.60 0.7703
0.70 0.7361
0.80 0.6969
0.90 0.6518
1.00 0.6000
1.10 0.5405
1.20 0.4722
1.30 0.3937
1.40 0.3036
1.50 0.2000
1.60 0.0810
1.661 0
1.70 0.0000
2.00 0.0000
};
\addplot+[mark=none, thick] table[x=R,y=D03] {
R D03
0.00 0.7000
0.10 0.6554
0.20 0.6041
0.30 0.5453
0.40 0.4777
0.50 0.4000
0.60 0.3108
0.70 0.2083
0.80 0.0906
0.8685 0
1.0 0
2.0 0
};
\addplot+[mark=none, thick] table[x=R,y=D05] {
R D05
0.00 0.5000
0.10 0.4257
0.20 0.3402
0.30 0.2421
0.40 0.1294
0.50 0.0000
0.6 0
2.0 0
};
\legend{$D=0.1$, $D=0.3$, $D=0.5$}
\end{axis}
\end{tikzpicture}
\caption{Outage probability \(P_{\rm out}(R,D)\) for the compact Gaussian class with (\(\sigma_0^2=1\)), as a function of rate budget \(R\), for several fixed distortion levels \(D\). Each curve reaches zero at \(R=\tfrac12\log_2(\sigma_0^2/D)\), consistent with \eqref{eq:gaussian_pout}.}
\label{fig:pout_gaussian}
\end{figure}
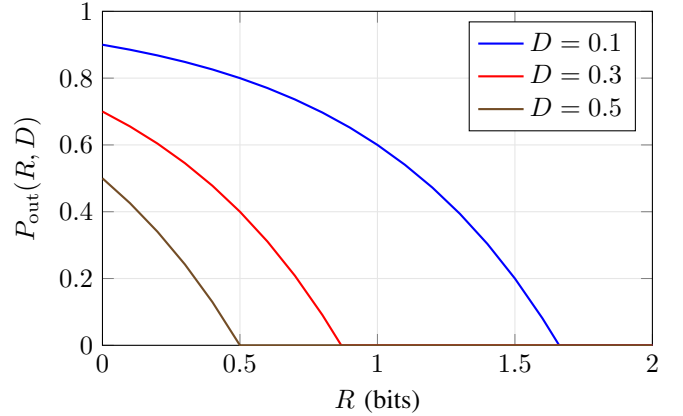

Since this class is a special case of the class of stable sources which have been recently studied in~\cite{karimRD}, the DROP expression can be readily extended to the general class of stable sources.

\end{example}

\section{Conclusion}
We revisited the source coding problem for a class of sources and introduced Distortion-Rate Outage Probability (DROP) as a performance metric for slow sources, in which the underlying source \(\theta\) remains fixed over the duration of transmission of the codeword. This characterization provides a natural source-side analogue of the channel-side notions of outage capacity and distortion outage. Through illustrative examples, we demonstrated how DROP highlights the tradeoff between outage probability, rate and distortion constraints. This formulation may further serve as a foundation for a joint outage notion in which a slow class of sources is transmitted over a slow-fading channel.

\section*{AI Usage}
The authors acknowledge the use of ChatGPT for English and \LaTeX-related assistance, and generation of Figures~\ref{fig:pout_bernoulli} and \ref{fig:pout_gaussian}. The main ideas and technical contents in this work were produced independently by the authors. Claude was used to review and refine some of the technical contents.

\bibliographystyle{IEEEtran}
\bibliography{ref}

\end{document}